\documentclass[10pt,aps,pre,twocolumn,superscriptaddress,floatfix,longbibliography,amsmath]{revtex4-1}
\usepackage{graphicx}
\usepackage{hyperref}
\usepackage{color}
\usepackage{float}
\usepackage{amssymb,amsthm}

\theoremstyle{definition}
\newtheorem*{proposition}{Proposition}
\newtheorem*{consequence}{Consequence}

\begin{document}
\title{Percolation and jamming of random sequential adsorption samples of large linear \texorpdfstring{$k$-mers}{k-mers} on a square lattice}

\author{M. G. Slutskii}
\email{mace\_window@mail.ru}
\affiliation{National Research University Higher School of Economics, 101000 Moscow, Russia}

\author{L. Yu. Barash}
\email[Corresponding author: ]{barash@itp.ac.ru}
\affiliation{National Research University Higher School of Economics, 101000 Moscow, Russia}
\affiliation{Landau Institute for Theoretical Physics, 142432 Chernogolovka, Russia}

\author{Yu. Yu. Tarasevich}
\email{tarasevich@asu.edu.ru}
\affiliation{Astrakhan State University, 414056 Astrakhan, Russia}

\begin{abstract}
The behavior of the percolation threshold and the jamming coverage for isotropic random sequential adsorption samples has been studied by means of numerical simulations. A parallel algorithm that is very efficient in terms of its speed and memory usage has been developed and applied to the model involving large linear $k$-mers on a square lattice with periodic boundary conditions. We have obtained the percolation thresholds and jamming concentrations for lengths of $k$-mers up to $2^{17}$. New large $k$ regime of the percolation threshold behavior has been identified. The structure of the percolating and jamming states has been investigated. The theorem of G.~Kondrat, Z.~Koza, and P.~Brzeski [Phys. Rev. E 96, 022154 (2017)] has been generalized to the case of periodic boundary conditions. We have proved that any cluster at jamming is percolating cluster and that percolation occurs before jamming.

\end{abstract}

\maketitle

\section{Introduction\label{sec:intro}}

Random sequential adsorption (RSA) is a standard method of modeling the adsorption of particles
at a liquid-solid interface~\cite{Talbot2000CSA}. Such particles can be, for example, polymers, biomolecules or nanotubes, their aspect ratios typically being in the orders of $10^3$--$10^4$. During deposition, a spanning cluster may occur, i.e., a set of neighboring particles which span between the opposite borders of the substrate. Additional adsorption of particles leads the system into a jammed state where no additional particle can be deposited due to the absence of appropriate holes~\cite{Evans1993RMP}. The resulting concentration of particles is known as the jamming coverage, $p_j$. The formation of spanning clusters is associated with the percolation phase transition, and the corresponding concentration of particles is known as the percolation threshold, $p_c$. Over the last several decades, percolation theory has been employed to study the properties of disordered media~\cite{Stauffer,Sahimi1994}. The properties of the media are considerably different below and above the percolation threshold. Some models of  percolation theory are simple, yet productive, models of phase transitions.

One of the possible ways to simulate the adsorption of such large elongated particles is based on the use of a discrete space, e.g., a square lattice and linear $k$-mers (also denoted as needles, linear segments, stiff-chains, rods, or sticks), i.e., rectangular ``molecules'', which occupy $k$ successive faces of the lattice.

Becklehimer and Pandey reported that, for $k \le 20$, the percolation threshold depends strongly on the value of $k$ with a power-law
\begin{equation}\label{eq:Becklehimer}
  p_c(k) \propto k^{-1/2},
\end{equation}
whereas the jamming coverage decreases~\cite{Becklehimer1992}.

The percolation of $k$-mers up to $k = 40$ has been studied by Leroyer and Pommiers~\cite{Leroyer1994PRB}. The investigators found that the percolation cluster is built of regions (stacks) of the same orientation, the typical size of which is $k$. Non-monotonic dependence of the percolation threshold on $k$ was found. For small values of $k$, the percolation threshold decreases as \begin{equation}\label{eq:Leroyer}
p_c(k) = k^{-1} + \mathrm{const}
\end{equation}
but it increases for larger values of $k$. Minimal values of the percolation threshold correspond to $k \approx 15$. The authors argued this behavior in terms of the stack structure of the percolation cluster. The authors suggested that the percolation threshold can be reached for any (finite) size of the $k$-mers in contrast to the simulations with rectangles and squares where jamming saturation occurs before percolation~\cite{Nakamura1987PRA}. Moreover, the structure of the percolation cluster was studied using the local order parameter
\begin{equation*}\label{eq:S}
  S = \left\langle\frac{\left|N_y - N_x\right|}{N_y + N_x}\right\rangle,
\end{equation*}
where $N_x$ ($N_y$) is the number of sites covered by the horizontal (vertical) needles in a given box.  Note that the stack structure of the jamming state in terms of the order parameter was studied in~\cite{Tarasevich2018MC}, where it was found that the characteristic size of each stack is of the order of $k\times k$.

It is evident that Eqs.~\eqref{eq:Becklehimer} and \eqref{eq:Leroyer} describe quite different behaviors. This is due to the use of different kinds of RSA, viz., Leroyer and Pommiers used conventional RSA, while Becklehimer and Pandey applied a so-called ``end-on'' mechanism of RSA~\cite{Evans1993RMP}.

The percolation and jamming of needles and squares on a square lattice have been investigated by Vandewalle et al.~\cite{Vandewalle2000epjb}. The ratio $p_c/p_j$ was found to be a constant of $0.62 \pm 0.01$ up to $k=20$.
The authors suggested that both quantities scale as
\begin{equation}\label{eq:Vandewalle}
  p(k) = C\left( 1 - \gamma\left( \frac{k - 1}{k}\right)^2 \right),
\end{equation}
where $\gamma  = 0.31 \pm 0.01$ (except for the case $k=1$). The presented results evidence that the percolation threshold is significantly smaller than the jamming coverage when $k \to \infty$ if Eq.~\eqref{eq:Vandewalle} is valid for any values of $k$. Moreover, the existence of stacks, not only at the percolation threshold but also in the jammed state was reported.

The jamming and percolation of needles with lengths up to $k=2000$ have been investigated by Kondrat and P\c{e}kalski~\cite{Kondrat2001PRE}. For $k>15$, an increase in the percolation threshold was found.
 For $15 \le k \le 45$, the fitting formula
\begin{equation}\label{eq:Kondrat}
  p_c/p_j \propto 0.50 + 0.13 \log_{10} k
\end{equation}
was proposed. A monotonic increase in $p_c /p_j$ holds over a wide range of values of $k$ even up to $k =2000$. The dependence of the jamming concentration on the value of $k$ was fitted by
\begin{equation}\label{eq:jamming}
  p_j(k) = p_j(\infty) +  a / k^\alpha,
\end{equation}
where $p_j(\infty) = 0.66 \pm 0.01$, $a = 0.44$, $\alpha = 0.77$. The parameters of the fitting formula were refined in Ref.~\cite{Lebovka2011PRE} as $p_j(\infty) = 0.655 \pm 0.009$, $a = 0.416$, $\alpha = 0.720 \pm 0.007$ based on simulations and scaling analysis up to $k=256$. It was also found that
\begin{equation}\label{eq:Tarasevichpc}
  p_c(k) = a_0 / k^{\alpha_0} + b \log_{10} k + c,
\end{equation}
where $a_0 = 0.36 \pm 0.02$, $\alpha_0 =0.81 \pm 0.12$, $b = 0.08 \pm 0.01$, and $c =0.33 \pm 0.02$~\cite{Tarasevich2012PRE}. For the ratio $p_c/p_j$, a fitting formula was proposed
\begin{equation}\label{eq:Tarasevichpcpj}
p_c/p_j = B \log_{10} k + C,
\end{equation}
where $B= 0.119 \pm 0.003$ and $C = 0.513 \pm 0.006$~\cite{Tarasevich2012PRE}.

The monotonic behavior of the percolation threshold as a function of $k$ for $k \le 15$ was reported and the following fitting formula was proposed
\begin{equation}\label{eq:Cornette}
  p_c(k) = p_c^\ast + \Omega \exp\left( -\frac{k}{\kappa}\right),
\end{equation}
$p_c^\ast = 0.461 \pm 0.001$,  $\Omega = 0.197 \pm 0.02$, and $\kappa = 2.775 \pm 0.02$~\cite{Cornette2003epjb}.

Naturally, jamming and percolation produced by means of the RSA of $k$-mers has been studied on other kinds of substrates, e.g., on triangular lattices~\cite{Budinski-Petkovic2012,Perino2017SSMTE}. A non-monotonic size dependence of the percolation threshold and decrease of the jamming coverage have also been observed on such  triangular lattices~\cite{Perino2017SSMTE}.

In a recent paper~\cite{Kondrat2017PRE}, a proof was presented that there will be a percolating cluster in any jammed configuration of nonoverlapping fixed-length, horizontal, or vertical needles on a square lattice. The theorem disproves the recent conjecture~\cite{Tarasevich2012PRE,Tarasevich2015PRE,Centres2015JSM} that in the random sequential adsorption of such needles on a square lattice, percolation does not occur if the needles are longer than some threshold value $k^*$, estimated to be of the order of several thousand.

We note that percolation is absent for a jammed system in a continuous model~\cite{Ziff1990JPA,Viot1992PA}.
Hence, these results imply that the qualitative behavior of continuous percolation cannot be approximated using a finite lattice model, regardless of the lattice size.

The theorem ensures that the ratio of the percolation threshold and jamming coverage $p_c/p_j$ is well defined for all needle lengths, but it does not predict the asymptotic behavior of this quantity for $k \to \infty$, which has not been reported up to now to the best of our knowledge. The theorem was proved only for the case of rigid boundaries,
but it is also interesting to consider periodic boundary conditions for this problem. The aim of this work is the direct numerical verification of the theorem and a determination of the character of the dependence of the percolation thresholds on $k$ for large values of $k$.

For this purpose, we developed a parallel algorithm that is very efficient in terms of speed and memory usage (Section~\ref{sec:methods}). This has allowed us to obtain the percolation and jamming concentrations for lengths of $k$-mer up to $2^{17}$ and to demonstrate the new large $k$ regime of their behavior (Section~\ref{sec:results}). The percolation threshold and the jamming coverage show very different behavior for $k \gtrsim 500$ compared to the behaviors for $k \le 512$ (Section~\ref{sec:conc}). We also generalize the results of~\cite{Kondrat2017PRE} for the case of periodic boundary conditions in Appendix~\ref{sec:appendix}.

\section{Methods\label{sec:methods}}
\subsection{Details of simulations}
Deposition of the $k$-mers onto a square lattice $L \times L$ with periodic boundary conditions, i.e., onto a torus, was performed using the RSA mechanism. Calculations were performed with $L=100k$ for $k\le 2^{14}$ and $L=1\,638\,400$ for $k>2^{14}$. For $k\le 2^{14}$ a single calculation was performed, while for $k>2^{14}$ two independent runs were carried out.

To estimate statistical uncertainty, we performed additional test for particular values of $k$ and $L$ with up to $200$ independent runs.

Scaling analysis was performed for some particular cases to estimate the finite-size effect, since the percolation threshold is expected to vary as~\cite{Leroyer1994PRB}
\begin{equation}\label{eq:pcscaling}
p_c(\infty) - p_c(L) \propto \left(\frac{L}{k}\right)^{-1/\nu},
\end{equation}
where the critical exponent in 2D is $\nu = 4/3$~\cite{Stauffer}, while the jamming coverage is expected to vary as~\cite{Vagberg2011PRE,Krapivsky2010RSA}
\begin{equation}\label{eq:pjscaling}
|p_j(\infty) - p_j(L)| \propto L^{-1}.
\end{equation}

The cluster size distribution was also studied. At the percolation threshold, the average number (per site) of clusters containing $s$ sites each is expected to obey the following relation
\begin{equation}\label{eq:CSD}
  n_s \propto s^{-\tau}
\end{equation}
for large values of $s$. Here, $\tau$ is the Fisher exponent, $\tau = 187/91$ for a percolation in 2D~\cite{Fisher1967PPF,Stauffer}.

To characterize the connectedness of the system under consideration, the relative number of interspecific contacts was used
\begin{equation}\label{eq:nxy}
  n_{xy}^\ast = \frac{n_{xy}}{n_x + n_y +n_{xy}},
\end{equation}
where $n_x$ ($n_y$) is the number of contacts between sites belonging to the horizontal (vertical) needles, and  $n_{xy}$ is the number of contacts between needles oriented in mutually perpendicular directions. In fact, this quantity shows how strong the stacks are connected with each other.  A configuration with strongly connected stacks implies the existence of percolation, while poorly connected stacks imply that percolation may not occur~\cite{Tarasevich2012PRE}.

\subsection{Deposition of \texorpdfstring{$k$-mers}{k-mers} onto the lattice}
Straightforward realization of the RSA has two major drawbacks in the final stage of the algorithm, viz.,
\begin{enumerate}
  \item due to the lack of space, many attempts will fail prior to each successful placement of an object;
  \item the program would not know whether jamming had occurred.
\end{enumerate}
Both these problems can be solved by generating two lists of viable sites (one for each orientation)
when the density is large enough, e.g., after percolation has occurred~\cite{Nord1991JSCS,Brosilow1991PRA,Evans1993RMP}. ``Viable'' sites are those, where a $k$-mer can be placed, considering the orientation. After the lists are filled, the coordinates should be taken from a random position in these arrays. Regardless of these coordinates still being viable or not, the chosen element should always be removed from the list.

\subsection{Cluster labeling}
Percolation transition occurs when a cluster wraps around the torus. Such a cluster is defined as a wrapping cluster. In order to check whether percolation is present, one should label all the clusters of $k$-mers. This task can be done using the Union-Find algorithm~\cite{Newman2000PRL,Newman2001PRE}. The idea is to label each occupied site of the lattice with the cluster's number (ID). The algorithm itself consists of two parts.
\begin{description}
\item[Find] when a new $k$-mer is added to the lattice we must find which clusters touch this $k$-mer.
\item[Union] if a $k$-mer connects different clusters, they should be merged into a single cluster.
\end{description}
The union part can be achieved by reassigning the cluster ID for all the sites of one of the clusters. It is reasonable always to choose the smaller clusters when such a need to change their ID is required, meaning  it is useful to have previously stored the size of every cluster.

\subsection{Checking for percolation}

Due to the periodic boundary conditions, percolation occurs when a cluster wraps around the torus. We are looking for a cluster which wraps around the torus in one or both directions. Finding a wrapping cluster can be performed using the Machta algorithm~\cite{Machta1996PRE,Newman2001PRE}. This requires choosing one site in each cluster (the root) and, for each site in this cluster, storing its $x$ and $y$ displacements from the root. These displacements can be negative and they should be updated each time different clusters are merged.

It is important to notice that a wrapping percolation can only occur if the new $k$-mer connects two parts of the same cluster (to wrap around the lattice). Each time this happens the algorithm should calculate $\left|x_1-x_2\right|+\left|y_1-y_2\right|$, where $x_1,y_1$ are the displacements in the first site and
$x_2,y_2$ are those in the second. If this expression is not greater than $2k$, then no percolation is present.
Otherwise, percolation has occurred.

\subsection{Memory issues}

In order for the algorithm to work properly, at least three integer numbers (12 bytes) need to be stored for each site, viz., the cluster ID and the displacements. Keeping in mind that the lattice can contain about $10^{12}$ sites, this leads to enormous memory consumption, hence, some serious optimization is required for large values of $k$. The method is used only for $k \ge 96$.

The idea is to store the cluster ID and the displacements for the whole $k$-mer rather than for each site.
Each site should store four bits:
\begin{description}
\item[1\textsuperscript{st} bit]  whether the site is occupied or not,
\item[2\textsuperscript{nd} bit]  orientation of the $k$-mer the site belongs to,
\item[3\textsuperscript{rd} bit]  whether the site is the first one (parent site) in the $k$-mer,
\item[4\textsuperscript{th} bit]  information about the $k$-mer.
\end{description}

Using the first three bits it is easy to find the parent of the $k$-mer. Since the $k$-mer is long enough ($k \ge 96$) it can be divided into several parts, as shown in Fig.~\ref{fig:memory}. While writing information into the sites the appropriate number should be decomposed into 32 bits and put into the fourth bit in each of the 32 corresponding sites. While reading information this same procedure has to be performed in reverse.
\begin{figure}[!htb]
\includegraphics[width=\columnwidth]{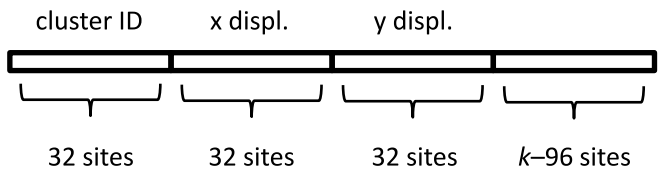}
\caption{Memory organization for storing the cluster ID and displacements.
Each site of every $k$-mer corresponds to a single bit of information.\label{fig:memory}}
\end{figure}

This method of memory usage requires $0.5$ byte for each site, which is $24$ times less than the original one. The method is used only for $k\ge 96$. Figure~\ref{fig:speedup}(a) shows that there is additional speedup due to the memory compression.
\begin{figure}[hbt]
\includegraphics[width=\columnwidth]{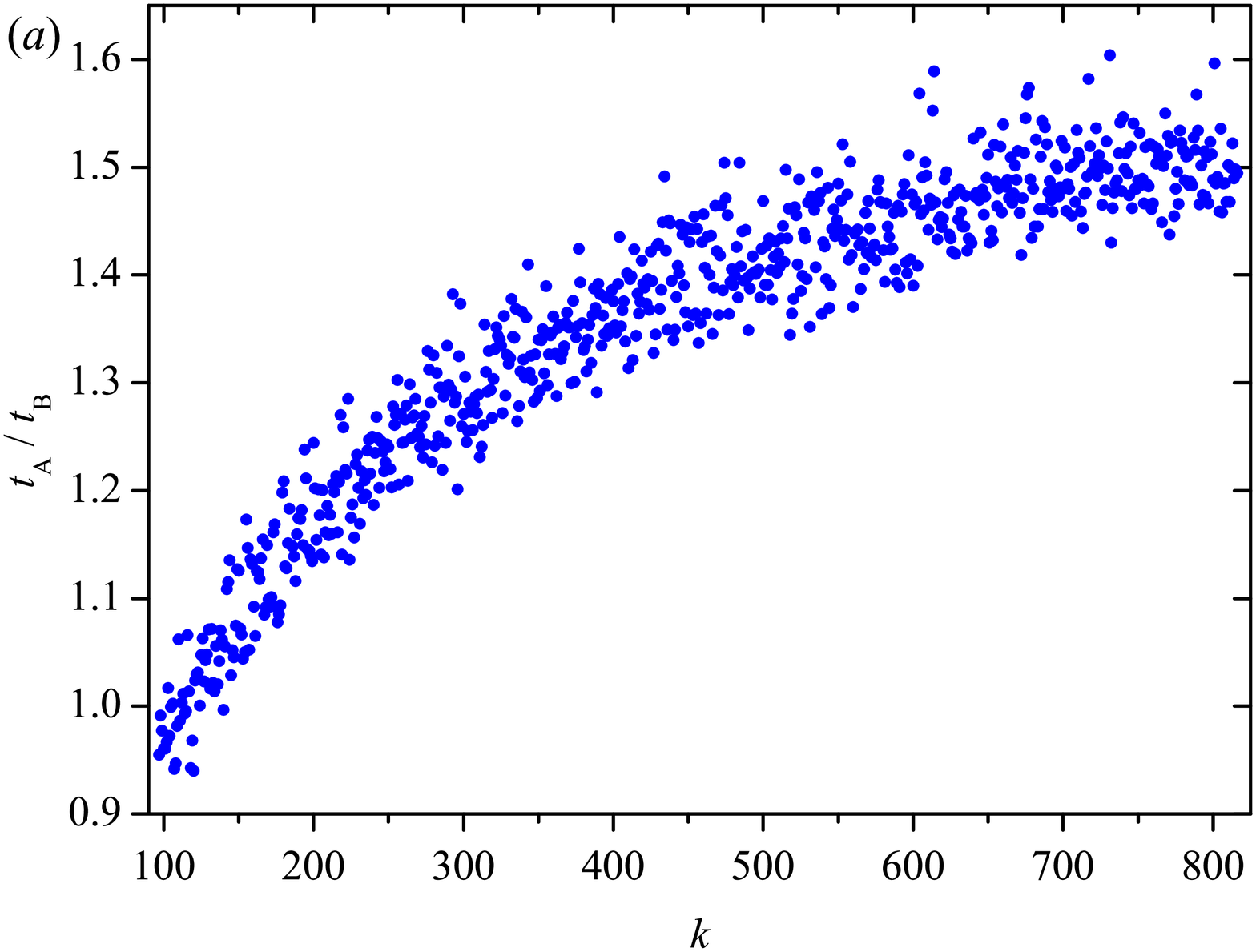}\\
\includegraphics[width=\columnwidth]{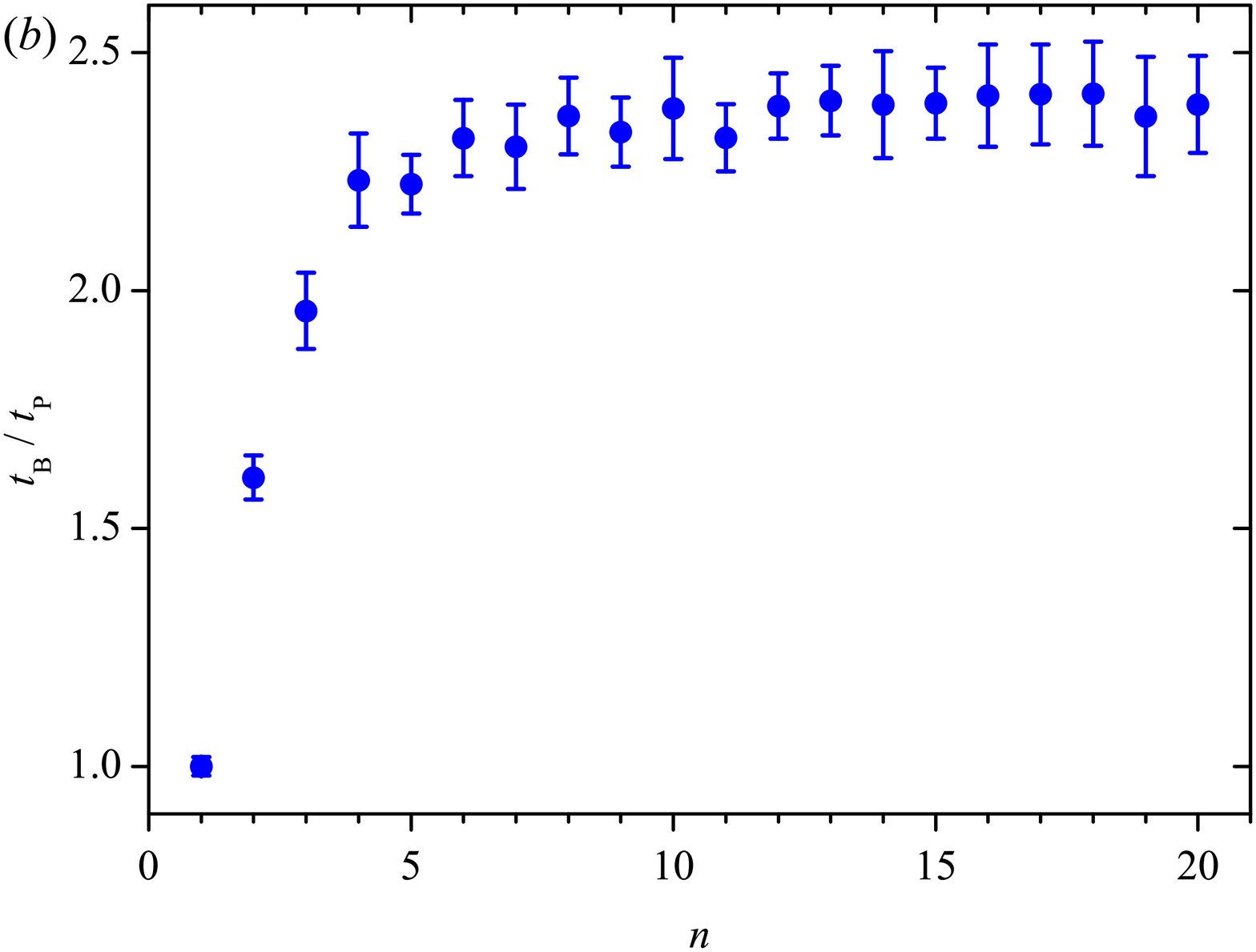}\\
\caption{(a) Dependence of $t_A/t_B$ on $k$, (b) Dependence of $t_B/t_P$ on $n$ for $k=512$. Here $t_A$ is running time of the sequential program without memory compression, $t_B$ is running time of the sequential program with memory compression, $t_P$ is running time of the parallel program with memory compression, which employs $n$ CPU cores.
\label{fig:speedup}}
\end{figure}

\subsection{Parallelization}

All simulations were performed on a computer with two Intel Xeon Platinum~8164, 2~GHz processors and 1536~GB RAM.
The speed of the algorithm is a matter of importance when it comes to modelling huge lattices.  This can be improved by using parallel programming with OpenMP.

One of the most time consuming parts is relabeling the clusters. It is obvious that different sets of clusters
can be merged simultaneously by parallel threads if they do not connect with each other.
To prevent different threads from processing the same clusters an additional shared array is used,
where each of the clusters can be marked as ``busy'' or ``free''.

To prevent intersections only one thread should be placing a $k$-mer at any time. This can be achieved using a critical section in OpenMP, which contains the following:
\begin{itemize}
\item check whether a $k$-mer can be placed, (otherwise the algorithm quits the critical section and randomly chooses other lattice site),
\item check that the surrounding clusters are ``free'', (otherwise the algorithm quits the critical section and waits for the surrounding clusters to be processed)
\item place the $k$-mer,
\item mark all surrounding clusters as ``busy''.
\end{itemize}

Prior to entering the critical section each thread makes sure that a $k$-mer can be placed, and that the surrounding clusters are not ``busy''. The thread waits if the latter condition does not hold. Verification of these conditions inside the critical section is necessary for the correctness of the program, while the same checkup outside the critical section is performed in order to minimize the number of entries to the critical section, because extra entries can significantly slow down the program. The clusters are marked ``free'' after all the necessary processes of merging the clusters and checking for percolation have been completed.

Figure~\ref{fig:speedup}(b) shows that the approach allows to speed up the program approximately by a factor of $2.4$. The speedup is limited to the condition that only one thread is placing a $k$-mer at any time, so the parallelization is efficient for a large number of CPU cores only in the final phase of the calculation, when a lot of cluster relabeling is being performed. Future improvements of the parallelization method may include dividing the lattice into several parts, and allowing several threads to place $k$-mers simultaneously into different parts of the lattice.

We have used the library RNGAVXLIB~\cite{Guskova2016CPC} for the efficient and safe generation of random numbers in parallel threads.

\section{Results and discussion\label{sec:results}}
\subsection{Percolating threshold, jamming coverage and their ratios}

The obtained percolation thresholds and jamming coverages and their ratios are shown in Fig.~\ref{fig:densities} and presented in Table~\ref{tab:values}. Known approximations~\eqref{eq:Vandewalle}, \eqref{eq:Kondrat}, \eqref{eq:Tarasevichpc}, and  \eqref{eq:Cornette} are also shown in Fig.~\ref{fig:densities} for comparison.
For $2^4 \le k \le 2^{17}$, the data can be fitted by
\begin{equation}\label{eq:fit}
  p_c = A + \frac{B}{ C + \sqrt{k} },
\end{equation}
where $A = 0.615 \pm 0.001$, $B = -2.26 \pm 0.09$, $C = 10.2 \pm 0.6$, the adjusted coefficient of determination is $R^2 =0.999$.
\begin{table}[H]
  \caption{Percolation thresholds $p_c$ and jamming coverages $p_j$  for different values of $k= 2^n$.\label{tab:values}}
 \begin{ruledtabular}
\begin{tabular}{ccc}
  $n$ & $p_c$ & $p_j$ \\
  \hline
0 & 0.5933 & 1\\
1 & 0.5611 & 0.9062\\
2 & 0.5023 & 0.8094\\
3 & 0.4717 & 0.7479\\
4 & 0.4590 & 0.7103\\
5 & 0.4733 & 0.6892\\
6 & 0.4894 & 0.6755\\
7 & 0.5111 & 0.6686\\
8 & 0.5292 & 0.6628\\
9 & 0.5450 & 0.6618\\
10 & 0.5628 & 0.6592\\
11 & 0.5740 & 0.6596\\
12 & 0.5850 & 0.6575\\
13 & 0.5913 & 0.6571\\
14 & 0.6003 & 0.6561\\
15 & 0.6039 & 0.6548\\
16 & 0.6086 & 0.6545\\
17 & 0.6067 & 0.6487\\
\end{tabular}
\end{ruledtabular}
\end{table}

\begin{figure}
\includegraphics[width=\columnwidth]{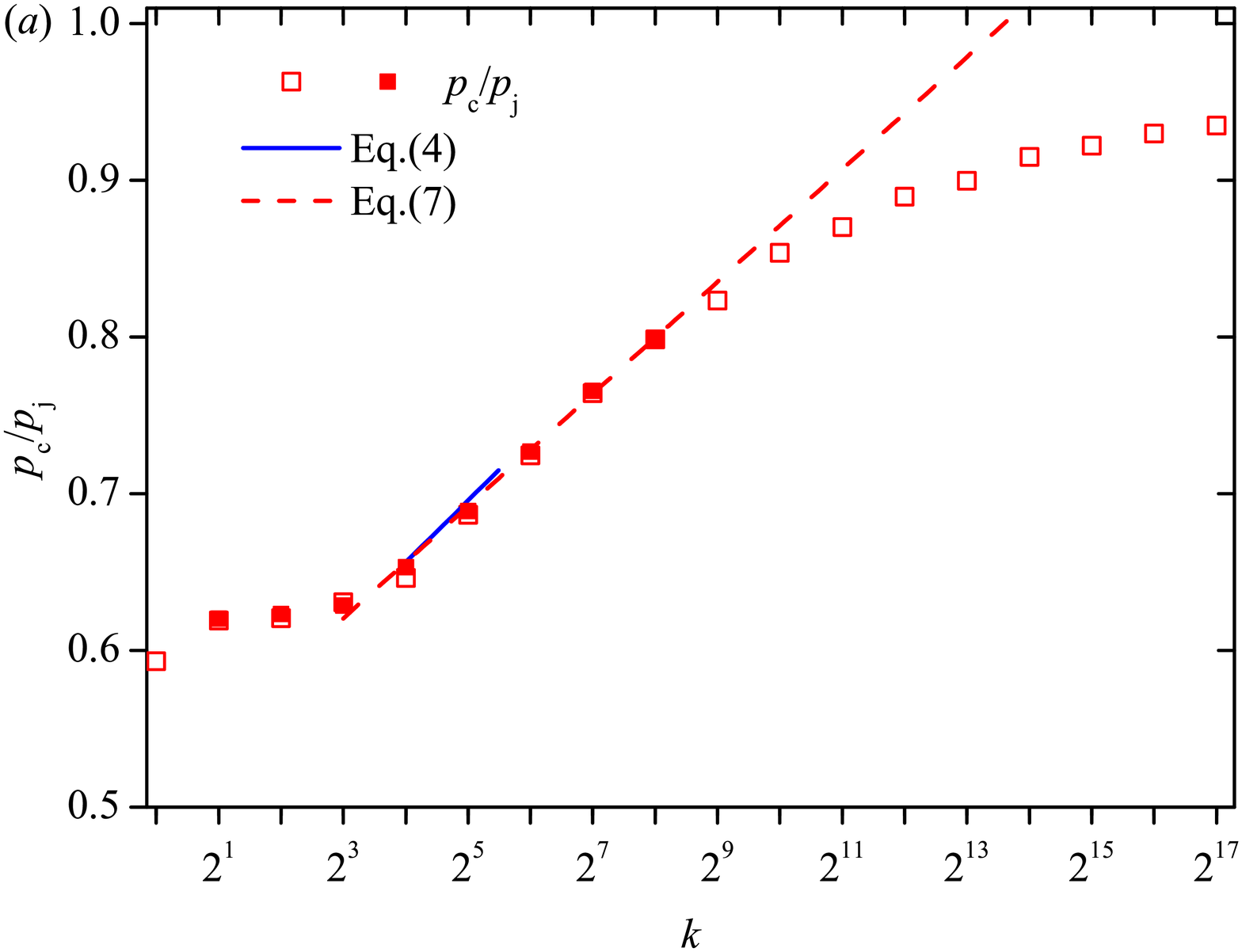}\\
\includegraphics[width=\columnwidth]{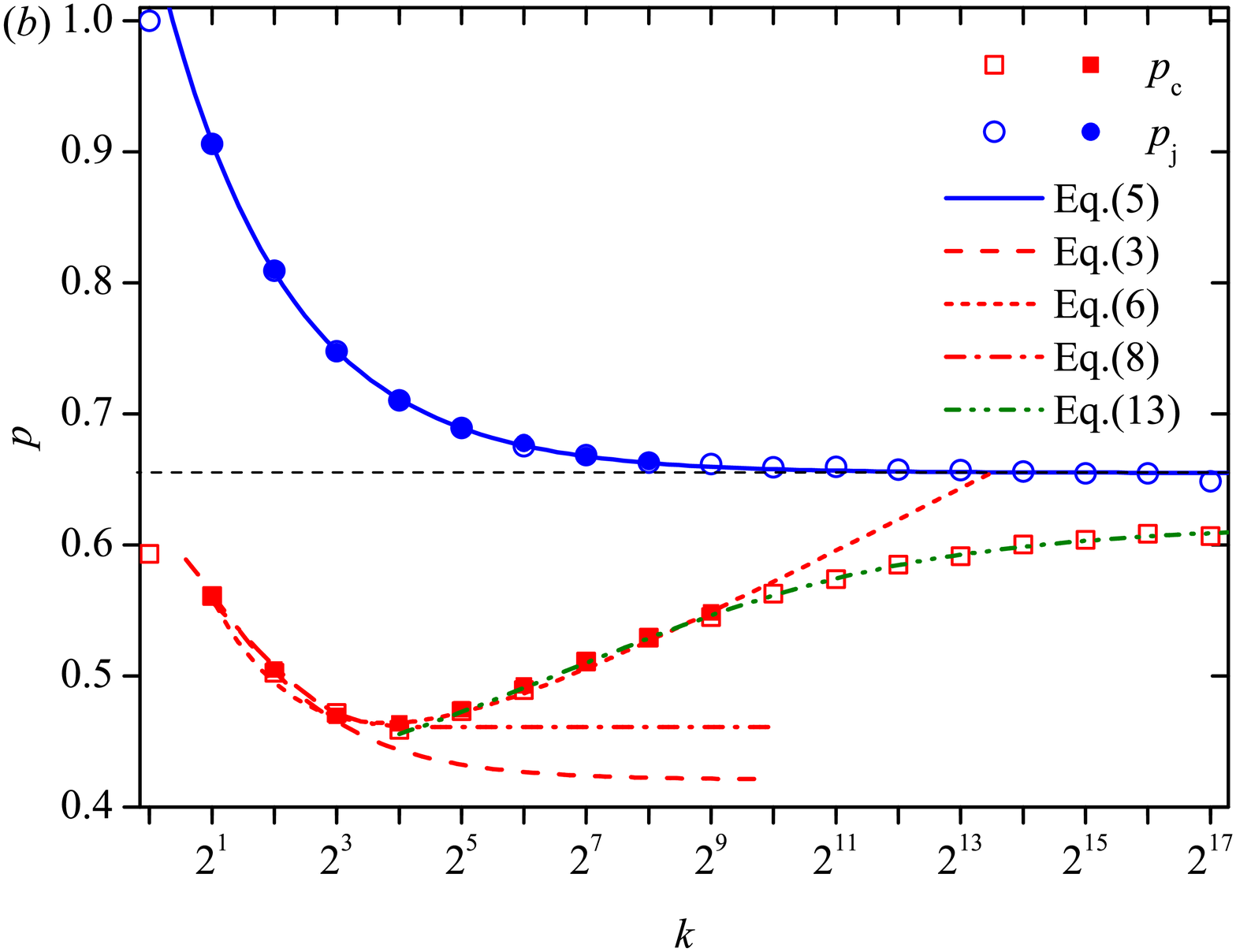}\\
\caption{(a) Dependence of $p_c/p_j$ on $k$. (b) Dependence of $p_c$ and $p_j$ on $k$.
Calculations were performed with $L=100k$ for $k\le 2^{14}$ and $L=1\,638\,400$ for $k>2^{14}$.
For $k\le 2^{14}$ a single calculation was performed, for $k>2^{14}$ two independent runs were carried out,
and the  mean values of $p_c$ and $p_j$ are shown. Previously known values of $p_c/p_j$ for $k\le 512$
taken from~\cite{Tarasevich2012PRE,Lebovka2011PRE} are shown as closed symbols. Parameters of the fitting curves are the same as indicated in the text.\label{fig:densities}}
\end{figure}

The results are in agreement with the previously known values of $p_c/p_j$ for $k\le 512$ taken from~\cite{Tarasevich2012PRE,Lebovka2011PRE}, which are shown as closed symbols. Figure~\ref{fig:densities} shows the new large $k$ regime of the behavior of the percolation threshold. In particular, the ratio $p_c/p_j$ shows very different behavior for $k\gtrsim 500$ compared to the behavior for $k\le 512$. The slow increase of the ratio for large values of $k$ compared to Eq.~\eqref{eq:Tarasevichpcpj} is in agreement with the results of~\cite{Kondrat2017PRE}. We have generalized the results of~\cite{Kondrat2017PRE} for the case of periodic boundary conditions in Appendix~\ref{sec:appendix} and have proved that, in thermodynamic limit, percolation always occurs before jamming.

\subsection{Estimations of the accuracy}

Figure~\ref{fig:fixk_p} shows the dependence of the standard deviations (STDs), $\sigma$, of $p_c$ and $p_j$ on $L/k$ for $k=256$, calculated over $100$ independent runs.
The STDs in Fig.~\ref{fig:fixk_p} decrease with $L$ approximately
as $\sigma(p_c)\propto L^{-3/4}$ and $\sigma(p_j) \propto L^{-1}$,
in accordance with Eqs.~\eqref{eq:pcscaling} and~\eqref{eq:pjscaling}.

The finite size effect is expected to be weak, except for the case when $k$ and $L$ are comparable~\cite{Manna1991JPA}. The STDs of the ratio $p_c/p_j$ is approximately $0.0051$ for $k=256$ and $L=100k$,
hence, a single calculation provides reasonably high accuracy for the case $L=100k$.
\begin{figure}[!htb]
\includegraphics[width=\columnwidth]{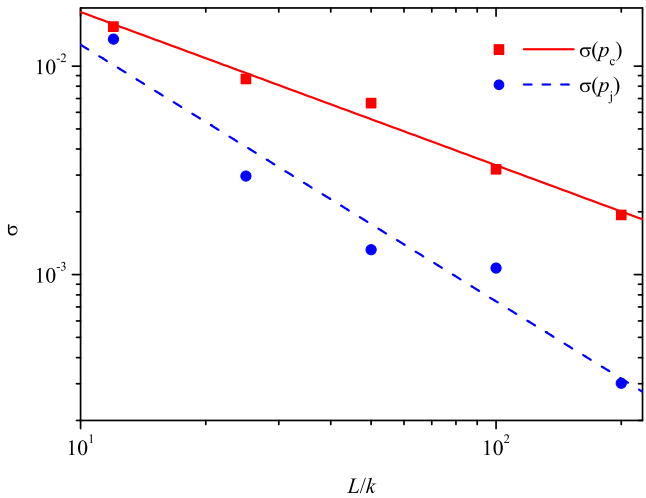}
\caption{STDs, $\sigma$, of $p_c$ and $p_j$ vs $L/k$ for $k=256$, estimated over $100$ independent runs.
The critical exponents are $0.74 \pm 0.05$ for the STD of the percolation threshold and $1.23 \pm 0.17$ for the STD of the jamming coverage.\label{fig:fixk_p}}
\end{figure}

Figures~\ref{fig:fixl_p} and~\ref{fig:fixldivk_p} show that the statistical error primarily depends on $L$, and
only weakly depends on $k$.
\begin{figure}[!htb]
\includegraphics[width=\columnwidth]{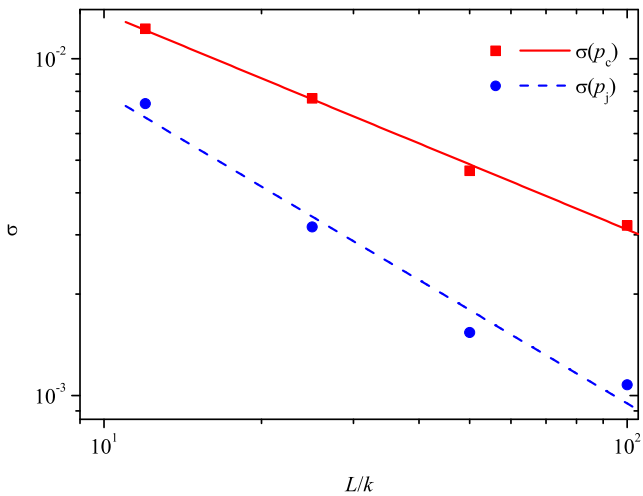}
\caption{STDs, $\sigma$, of $p_c$ and $p_j$ vs $L/k$ for $L=25600$, estimated over $200$ independent runs.
The critical exponents are $0.643 \pm 0.024$ for $p_c$ and $0.92 \pm 0.11$ for $p_j$.
\label{fig:fixl_p}}
\end{figure}
\begin{figure}[!htb]
\includegraphics[width=\columnwidth]{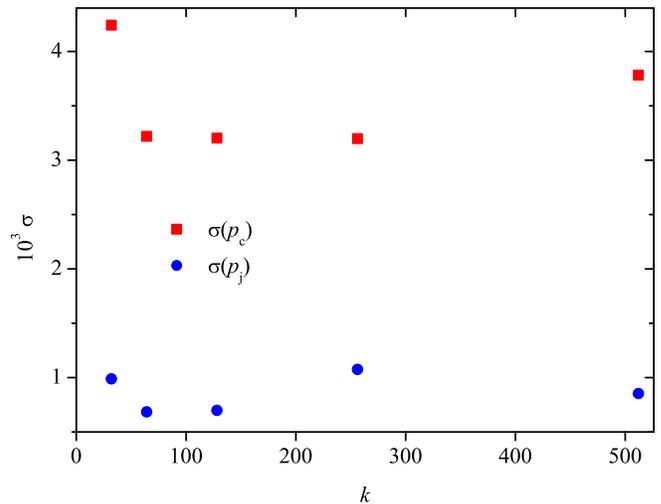}
\caption{Standard deviations of $p_c$ and $p_j$ vs $k$ for $L=100k$, estimated over $200$ independent runs.
\label{fig:fixldivk_p}}
\end{figure}

\subsection{Structure of the percolating and jamming states}

Figure~\ref{fig:clusterstat} shows the cluster size distribution at the percolating threshold for $k=256$, $L=100k$.
$n_s$ is the average number per site of clusters consisting of $s$ sites each.
The dashed line shows the function $s^{-187/91}$, so one can see that the distribution
behaves approximately as $n_s\propto s^{-187/91}$  in accordance with Eq.~\eqref{eq:CSD}.
\begin{figure}[tb]
\includegraphics[width=\columnwidth]{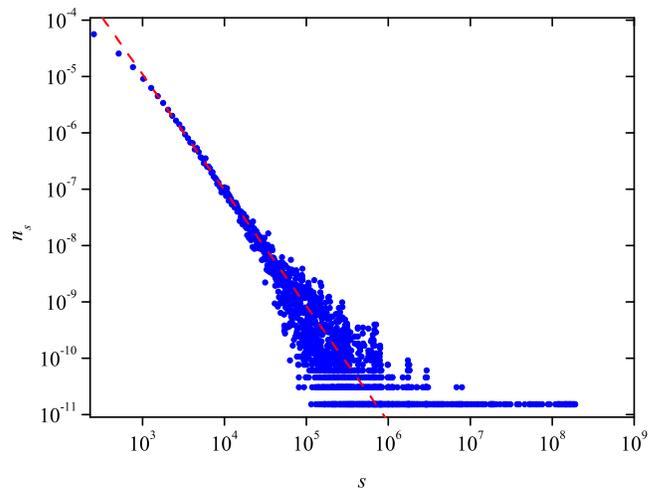}
\caption{Cluster size distribution for $k=256$, $L=100k$, at the percolation threshold, averaged over 100 independent runs. $n_s$ is the number of clusters consisting of $s$ sites divided by the total number of sites, $L^2$. The distribution behaves approximately as $n_s\propto s^{-187/91}$ (dashed line).\label{fig:clusterstat}}
\end{figure}

We found that, in practice, the jamming state always contains only a single cluster. However, it is theoretically possible to construct a jamming state with several clusters. An example is shown in Fig.~\ref{jamming2clusters} for $k=4$ and $L=10k$.
\begin{figure}[!htb]
\includegraphics[width=0.5\columnwidth]{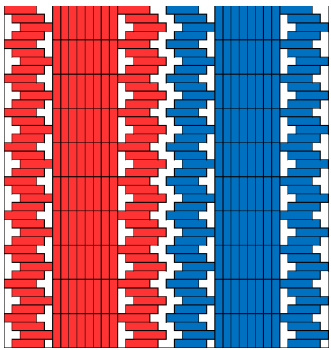}
\caption{Example of a jamming state with two clusters for a square lattice with $k=4$, $L=10k$
and periodic boundary conditions. Courtesy of R.\,K.\,Akhunzhanov.}
\label{jamming2clusters}
\end{figure}

Figure~\ref{fig:contacts} shows that the relative number of interspecific contacts decreases with $k$
as $n^\ast_{xy} \propto k^{-1.82}$. The exponent $-1.82$ is larger than $-2$. For a configuration with an average stack size $k\times k$ and a ragged interface between stacks, this means that the stacks are connected with each other providing percolation. Also, Fig.~\ref{fig:contacts} demonstrates a significant finite size effect for $k \gtrapprox 2^8$.
The standard deviation of the obtained values of $n^*_{xy}$ estimated for $k=2^8$ and $L=100k$ over $200$ independent runs is $\sigma(n^*_{xy}) \approx 3.1\times 10^{-6} \approx 0.014 n^*_{xy}$.
\begin{figure}[!htb]
  \centering
  \includegraphics[width=\columnwidth]{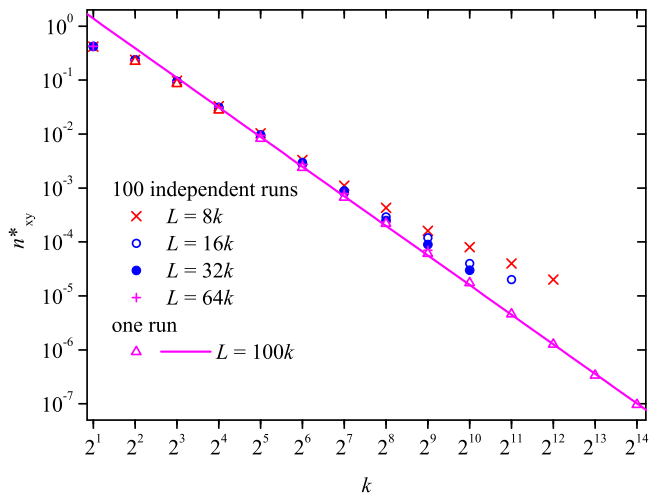}
  \caption{Relative number of interspecific contacts, $n_{xy}^\ast$, vs $k$ at jamming. Solid line corresponds to the  fit $ n^\ast_{xy} \propto k^{-1.82}$.\label{fig:contacts}}
\end{figure}

Figure~\ref{fig:percjam} demonstrates the stack structure of a wrapping cluster. The stacks of mutually perpendicular orientations are connected according to the behavior of the relative number of interspecific contacts.  The edges of the clusters are ``ragged''. Hence, the stacks look like clouds with a typical size of $k \times k$ and have rather diffuse edges, not like clear-cut squares. Internal regions of the stacks contain holes even at jamming.
\begin{figure}[!htb]
  \centering
  \includegraphics[width=\columnwidth,clip=true]{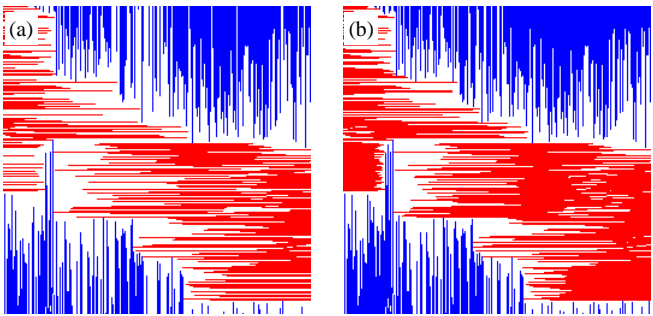}
  \caption{Fragment ($256 \times 256$) of a wrapping cluster at the percolation threshold (a) and the same fragment at jamming (b). $k=128$, $L= 100k$.\label{fig:percjam}}
\end{figure}

\section{Conclusion\label{sec:conc}}

We have studied the behavior of percolation and jamming thresholds for the isotropic random sequential adsorption of large linear $k$-mers onto a square lattice. We have presented a parallel algorithm which is very efficient in terms of speed and memory usage. Our results are in agreement with the previously obtained results  of~\cite{Tarasevich2012PRE,Lebovka2011PRE,Kondrat2017PRE} and we have obtained the percolation and jamming concentrations for lengths of $k$-mer up to $2^{17}$. The ratios of percolation and jamming densities show quite different behavior for $k \gtrsim 500$ compared to the behavior for $k \le 512$. We are not aware of any other studies of this problem that have considered values of $k$ larger than $512$ except~\cite{Kondrat2001PRE} where the values of $k$ were up to 2000 and the lattice size was only $L = 2500$. We have also analyzed the structure of the percolating and jamming states in terms of the cluster size distribution and the relative number of interspecific contacts.

\acknowledgments
The simulations were performed on the HPC facilities of the Science Center in Chernogolovka. We acknowledge the funding from the Basic Research Program of the National Research University Higher School of Economics (M.G.S, L.Y.B.) and the Ministry of Education and Science of the Russian Federation, Project No.~3.959.2017/4.6 (Y.Y.T.). The authors would like to thank Grzegorz Kondrat for his stimulating discussions.

\appendix

\section{Jammed systems and periodic boundary conditions\label{sec:appendix}}

\begin{proposition}
Each cluster in every jammed configuration of fixed-length nonoverlapping horizontal or vertical needles
on a finite square lattice with periodic boundary conditions is a percolating cluster.
\end{proposition}

\begin{proof}
Consider the sides of an arbitrary square of size $L$ as the edges of the system, where $L\times L$ is the lattice size. It has previously been proved, for rigid boundary conditions, that every cluster in a jammed configuration
extends to one of two consecutive edges of the system (see the Lemma in Method~I in~\cite{Kondrat2017PRE}).
This statement holds also for periodic boundary conditions, because its proof in~\cite{Kondrat2017PRE}
is directly applicable for this case.
We will label the system edges using geographical notation
($N$, $E$, $S$, and $W$ for the top, right, bottom, and left edge, respectively).

There are two cases: either the cluster extends to all four edges of the system, or it does not. In the latter case, it does not touch at least one system edge, say, $N$. Then it follows from the Lemma that it extends to the two edges adjacent to $N$, that is, to $E$ and $W$. Let us now shift the edges to the left by one column, while keeping the configuration unchanged. Then the cluster still does not touch the edge $N$, so it must still touch the new edge $E$ in at least one lattice site. Proceeding with such shifts to the left, each time we find, in the leftmost column, at least one lattice site, that belongs to the cluster. We proceed with such shifts until the same lattice site is found twice during this process. This will happen eventually, because there are only a finite number of lattice sites. On the other hand, finding the same site twice during this process means that the cluster wraps all the way around the lattice, hence it is a percolating cluster.

In the former case, when the cluster touches all four edges of the system, we shift the edges upward by one row.
We proceed with such shifts until the cluster does not touch the edge $N$. If the process completes and the cluster
does not touch the edge $N$, then such a situation was considered in the previous paragraph, so it is a percolating cluster. If the process does not complete during $L$ shifts, then we find at each step a lattice site where the cluster touches the edge $N$. We proceed with such shifts until the same lattice site is found twice during this process. This will happen eventually, because there are only a finite number of lattice sites. On the other hand, finding the same site twice during this process means that the cluster wraps all the way around the lattice, hence it is a percolating cluster.
\end{proof}

\begin{consequence}
In a system of fixed-length non-overlapping horizontal or vertical needles on a square lattice with periodic boundary conditions (on a torus), in the thermodynamic limit, percolation always occurs before jamming, $p_c < p_j$.
\end{consequence}

\begin{proof}
Pinson proved~\cite{Pinson1994JSP} that the probability, $R$, of finding a wrapping cluster on a torus at the percolation threshold is $R(p_c) = R^\ast$, where $0 < R^\ast < 1$. The value of $R^\ast$ depends on the way how the wrapping cluster is found. In particular, in the thermodynamic limit, the probability that there exists a cluster that wraps around the boundary conditions in the horizontal direction is $R^\ast = 0.521\,058\,290$~\cite{Newman2000PRL,Newman2001PRE}. For $L \to \infty$,
\begin{equation}\label{eq:R}
  R =
  \begin{cases}
    0, & \text{if } p<p_c, \\
    R^\ast, & \text{if } p=p_c, \\
    1, & \text{if } p>p_c.
  \end{cases}
\end{equation}
Since $R(p_j) = 1$ therefore $p_j >p_c$.

\end{proof}

\bibliography{percolation}

\end{document}